\documentclass[a4paper,12pt]{article}
\usepackage[utf8]{inputenc}
\usepackage{amsmath, amssymb, amsthm, comment, hyperref, cite, graphicx, authblk}
\usepackage{float}

\newtheorem{theorem}{Theorem}

\usepackage{anysize}
\marginsize{2.0cm}{2.0cm}{2.0cm}{2.0cm}
\begin{document}

\title{General Representation of Nonlinear Green's Function for Second Order Differential Equations Nonlinear in the First Derivative}

\author[1]{Marco Frasca}

\author[2, 3]{Asatur Zh. Khurshudyan\footnote{Email: khurshudyan@mechins.sci.am}}

\affil[1]{\small Via Erasmo Gattamelata 3, 00176 Roma, Italy}
\affil[2]{\small Department on Dynamics of Deformable Systems and Coupled Fields, Institute of Mechanics, National Academy of Sciences of Armenia, Yerevan, Armenia}
\affil[3]{\small Institute of Natural Sciences, Shanghai Jiaotong University, Shanghai, China}

\date{}

\maketitle

\begin{abstract}
The Green's function method which has been originally proposed for linear systems has several extensions to the case of nonlinear equations. A recent extension has been proposed to deal with certain applications in quantum field theory. The general solution of second order nonlinear differential equations is represented in terms of a so-called short time expansion. The first term of the expansion has been shown to be an efficient approximation of the solution for small values of the state variable. The proceeding terms contribute to the error correction. This paper is devoted to extension of the short time expansion solution to non-linearities depending on the first derivative of the unknown function. Under a proper assumption on the nonlinear term, a general representation for Green's function is derived. It is also shown how the knowledge of nonlinear Green's function can be used to study the spectrum of the nonlinear operator. Particular cases and their numerical analysis support the advantage of the method. The technique we discuss grants to obtain a closed form analytic solution for non-homogeneous non-linear PDEs so far amenable just to numerical solutions. This opens up the possibility of several applications in physics and engineering.

\end{abstract}

\section{Introduction}	

One of the most common methods of analysis of non-homogeneous \emph{linear} differential equations is the Green's function method. It allows to obtain an explicit representation for the solution to a boundary value problem knowing its Green's function. If the general solution of the problem
\[
\mathcal{D}\left[G\right] = \delta\left(\boldsymbol x\right) + \delta g, ~~ \boldsymbol x \in \Omega \subseteq \mathbb{R}^n,
\]
\[
\mathcal{B}\left[G\right] = 0, ~~ \boldsymbol x \in \partial\Omega,
\]
is known, then the general solution of the problem
\[
\mathcal{D}\left[w\right] = f\left(\boldsymbol x\right), ~~ \boldsymbol x \in \Omega \subseteq \mathbb{R}^n,
\]
\[
\mathcal{B}\left[w\right] = g ~~ {\rm on} ~~ \partial\Omega,
\]
is defined as the convolution of $G$ and non-homogeneous part of the latter problem:
\begin{equation}\label{Greenform}
w\left(\boldsymbol x\right) = \int_{\Omega} \left[ f\left(\boldsymbol \xi\right) + \delta g \right] G\left(\boldsymbol x - \boldsymbol \xi\right) {\rm d}\boldsymbol \xi.
\end{equation}
Here $\mathcal{D}\left[\cdot\right]$ and $\mathcal{B}\left[\cdot\right]$ are state and boundary linear operators respectively, $f$ and $g$ are given, $\delta g$ is linear in $g$ and denotes the inclusion of the non-homogeneity in boundary conditions to the right-hand side of the state equation (see \cite{Avetisyan2018} for details).

In the derivation of (\ref{Greenform}), the superposition principle is used, hence making it to be applicable merely to \emph{linear} equations. Nevertheless, the Green's function method has been extended to nonlinear equations by introducing the so-called backward and forward propagators \cite{Cacuci1988, Cacuci1989} and by using the short time expansion of the solution in terms of the nonlinear Green's function \cite{Frasca2007, Frasca2008} named in the following NLG method. Both extensions provide low-error approximation for the main types of nonlinear differential equations. However, for second order nonlinear differential equations the NLG method has a simpler representation formula. Recently, the NLG method has been modified in \cite{Khurshudyan2018, Khurshudyan2018AMP, Khurshudyan2018IJMPC} and new integrable cases have been established. Moreover, it has been shown that the NLG method can be efficiently combined with the method of generalized variables separation to approximate nonlinear partial differential equations. Using the NLG method, the controllability of oscillating nonlinear processes has been studied in \cite{Avetisyan2017, Avetisyan2018}.

In this article we consider new types of nonlinear partial differential equations that can be approximated using the NLG method. The principal difference of this paper with respect to \cite{Frasca2007, Frasca2008, Khurshudyan2018, Khurshudyan2018AMP, Khurshudyan2018IJMPC} is that here the nonlinear term depends on the first derivative of the unknown function. The idea is to extend this technique to larger classes of equations that have wide applications making it possible to find new solutions, in a closed analytic form, describing completely new physical regimes, mostly in presence of a source term.

For the NLG technique to be applicable, one should be able to solve the corresponding equation for the non-linear Green function. Anyway, a numerical algorithm can be otherwise devised from this approach to work. It is interesting to notice how a large set of well-known non-linear PDEs can be treated in this way obtaining both the explicit solutions and the spectrum of the non-linear operator. This makes the NLG method an innovative and powerful technique to attack and solve explicitly non-linear non-homogeneous PDEs. Our aims in this paper is to enlarge further the class of equations amenable to this treatment.

\section{The NLG method and Nonlinear Green's Functions for Particular Non-Linearities}

It has been established in \cite{Frasca2007, Frasca2008} that the solution of the second order nonlinear ODE
\begin{equation}\label{gennonlin}
\frac{d^2 w}{d t^2} + N\left(w, t\right) = f\left(t\right), ~~ t > 0,
\end{equation}
with a generic non-linearity $N$ and a given source function $f$, admits the following short time expansion of the general solution
\begin{equation}\label{shorttime}
w\left(t\right) = a_0 \int_0^t G\left(t - \tau\right) f\left(\tau\right) {\rm d}\tau + \sum_{k = 1}^\infty a_k \int_0^t \left(t - \tau\right)^k G\left(t - \tau\right) f\left(\tau\right) {\rm d}\tau,
\end{equation}
where unknowns $a_k$, $k = 0, 1, 2, \dots,$ are determined in terms of the quantities $w^{\left(k\right)}\left(0\right)$.

Here $G$ is the solution of the differential equation
\begin{equation}\label{Greeneq}
\frac{d^2 G}{d t^2} + N\left(G, t\right) = s \delta\left(t - \tau\right),~~ t > 0,~ \tau > 0,
\end{equation}
under corresponding Cauchy conditions, where $\delta$ is the Dirac distribution. Due to the similarity with the linear case, $G$ is referred to as the nonlinear Green's function of (\ref{gennonlin}).

We have also shown \cite{Frasca2007, Frasca2008, Khurshudyan2018, Khurshudyan2018AMP, Khurshudyan2018IJMPC} that the first order term in the short time expansion above, i.e.,
\begin{equation}\label{Frascaapprox}
w\left(t\right) \approx a_0 \int_0^t G\left(t - \tau\right) f\left(\tau\right) {\rm d}\tau,
\end{equation}
provides a numerical approximation for several partial differential equations consistent with the numerical solution obtained by the well-known method of lines. At this, $a_0$ and $s$ are real parameters that must be chosen to minimize the approximation error. Note that $s$ has been introduced in \cite{Frasca2007, Frasca2008}, while $a_0$ has been introduced in \cite{Khurshudyan2018, Khurshudyan2018AMP, Khurshudyan2018IJMPC}.

The exact determination of the nonlinear Green's function strongly depends on the form of $N$. Two particular cases have been considered in \cite{Frasca2007}. Specifically, it has been shown that the cubic non-linearity
\[
N\left(w, t\right) = w^3,
\]
admits the following exact solution of (\ref{Greeneq}):
\[
G\left(t\right) = 2^{\frac{1}{4}} \theta\left(t\right) \cdot \operatorname{sn}\left[\frac{t}{2^{\frac{1}{4}}},i\right],
\]
where $\theta$ is the Heaviside function, and $\operatorname{sn}$ is the Jacobi snoidal function defined as follows:
\[
\operatorname{sn} \left(\sigma, m \right) = \sin \varphi, ~~ {\rm where} ~~ \sigma = \int_0^\varphi \frac{{\rm d}\phi}{\sqrt{1 - m \sin^2\phi}}.
\]

On the other hand, the trigonometric non-linearity
\[
N\left(w, t\right) = \sin w
\]
leads to a Green's function of the form
\begin{equation}\label{eq:am}
G\left(t\right) = 2\theta\left(t\right) \cdot \operatorname{am}\left[\frac{t}{\sqrt{2}}, \sqrt{2}\right],
\end{equation}
where $\operatorname{am}$ is the Jacobi amplitude function defined as follows:
\[
\operatorname{am} \left(\sigma, m \right) = \varphi, ~~ {\rm where} ~~ \sigma = \int_0^\varphi \frac{{\rm d}\phi}{\sqrt{1 - m \sin^2\phi}}.
\]
Some new particular cases have been also considered recently in \cite{Khurshudyan2018, Khurshudyan2018AMP, Khurshudyan2018IJMPC}. In Tab.~\ref{tab1}, we present some of the known explicitly integrable cases so far.
\begin{table}[H]
\centering
\begin{tabular}{ c c }
 Nonlinear term & Green's function \\ \hline
 
 $w^2$ & $\displaystyle -\frac{1}{c} \theta\left(t\right) \wp\left(c t + c_1; 0, c_2\right)$ \\[0.5cm] \hline
 
 $w + w^3$ & $\displaystyle \sqrt{2 - c_1} i ~ \theta\left(t\right) \operatorname{sn}\left[ \frac{\sqrt{c_1}}{\sqrt{2}} \left| t + c_2 \right|, \frac{2 - c_1}{c_1} \right]$ \\[0.5cm] \hline
 
 $\displaystyle \frac{1}{w}$ & $\displaystyle c_1 \theta\left(t\right) \exp\left[ - \varphi^2\left(t; c_1, c_2\right) \right]$ \\[0.5cm] \hline
 
 $\displaystyle \frac{1}{w^3}$ & $\displaystyle \theta\left(t\right) \cdot \frac{1}{\sqrt{c_1}} \sqrt{c_1^2 \left( t + c_2 \right)^2 - 1}$ \\[0.5cm] \hline
 
 $\exp w$ & $\displaystyle \theta\left(t\right) \cdot \ln\left[\frac{1}{2}c_1\left(1 - \tanh^2\left[\frac{1}{2}\sqrt{c_1 \left(t + c_2\right)^2}\right]\right)\right]$ \\[0.5cm] \hline
 
 $\cos w$ & $\displaystyle 2\theta\left(t\right) \cdot \operatorname{am}\left[\frac{t}{\sqrt{2}}, \sqrt{2}\right]-\frac{\pi}{2}$ \\[0.5cm] \hline
 
 $\sinh w$ & $\displaystyle 2i ~ \theta\left(t\right) \cdot \operatorname{am}\left[c_1 \left|t + c_2\right|, \frac{1}{c_1}\right]$ \\[0.5cm] \hline
 
 $\cosh w$ & $\displaystyle 2i ~ \theta\left(t\right) \cdot \operatorname{am}\left[c_1 \left|t + c_2\right|, \frac{1}{c_1}\right]-i\frac{\pi}{2}$ \\[0.5cm] \hline
 
 $\displaystyle w \frac{d w}{dt}$ & $\displaystyle c_1 \theta\left(t\right) \cdot \tanh\left[\frac{1}{2}c_1\left(t + c_2\right)\right]$ \\[0.5cm] \hline
 
 $\displaystyle w \left(\frac{d w}{dt}\right)^2$ & $\displaystyle -\sqrt{2}i ~ \theta\left(t\right) \cdot \operatorname{erf}^{-1}\left[\sqrt{\frac{2}{\pi}} i ~ c_1 \left( t + c_2 \right) \right]$ \\[0.5cm] \hline
 
 $\displaystyle g\left(w\right) \left(\frac{d w}{d t}\right)^3$ & $\displaystyle G\left(t\right) = \theta\left(t\right) \cdot G_0\left(t\right), ~ \left[ \int_0^{G_0} \left( c_2 + \int_0^z g\left(\zeta\right) {\rm d}\zeta \right) ~ {\rm d}z \right]^{-1}\left(t + c_1\right)$ \\[0.5cm] \hline
 
 $\displaystyle g\left(w\right) = w^3$ & $\displaystyle \theta\left(t\right) \cdot \frac{t + c_2 - c_1 W\left(\psi\left(t; c_1, c_2\right)\right)}{c_1}$ \\[0.5cm] \hline
\end{tabular}
\caption{Explicit nonlinear Green's functions}
\label{tab1}
\end{table}
In Tab. \ref{tab1}, $c_1$ and $c_2$ are integration constants determined from homogeneous Cauchy conditions and
\begin{enumerate}
\item $c = \left(-6\right)^{-\frac{1}{3}}$, $\wp$ is the Weierstrass elliptic function \[
\wp\left(t; \omega_1, \omega_2\right) = \frac{1}{t^2} + \sum_{n^2 + m^2 \neq 0}\left[\frac{1}{\left(t + \omega_1 m + \omega_2 n\right)^2} - \frac{1}{\left(\omega_1 m + \omega_2 n\right)^2}\right],
\]
\item $\operatorname{erf}^{-1}$ is the inverse of the Gauss error function
\[
\operatorname{erf}\left(t\right) = \frac{2}{\sqrt{\pi}}\int_0^t \exp\left[-\tau^2\right] {\rm d}\tau.
\] \[
\varphi\left(t; c_1, c_2\right) = \operatorname{erf}^{-1}\left[-\sqrt{\frac{2}{\pi}} \left|c_1\right|\left|t + c_2\right|\right],
\]
\item $g$ is any function for which the inverse function of the integral exists \cite{Bartlett2018}, $W$ is the Lambert function and 
\[
\psi\left(t; c_1, c_2\right) = \frac{1}{c_1} \exp\left[\frac{t - \tau + c_2}{c_1}\right].
\]
\end{enumerate}
 
From Tab~\ref{tab1}, we can recognize several well-known non-linearities entering in a lot of physical problems, also in fundamental ones like the cubic non-linearity or the Liouville non-linearity due the exponential. In both cases we are non in a position to provide an explicit solution for the non-homogeneous PDEs. This aspect is fundamental to treat these problems in quantum field theory in an explicit way (see e.g. \cite{Frasca:2005sx}).

\section{Green's functions and spectra}

The knowledge of the Green's function for some of these nonlinear equations can yield a precise information on their spectral properties. In fact, the poles of the Green's function yields directly the spectrum provided we can consider possible any further correction as a higher order effects. This is generally a quite acceptable condition. For some of the linearities considered in this paper, some results were obtained in \cite{Frasca2007}.

The technique can be worked out as follows. Let us consider the case with $w^3$. We have
\[
G\left(t\right) = 2^{\frac{1}{4}} \theta\left(t\right) \cdot \operatorname{sn}\left[\frac{t}{2^{\frac{1}{4}}}, i\right].
\]
This solution admits a Fourier series given by
\[
\operatorname{sn}\left(u\right) = \frac{2\pi}{i K\left(-1\right)} \sum_{n = 0}^\infty \left(-1\right)^n \frac{\exp\left[-\left( n + \frac{1}{2} \right) \pi\right]}{1 + \exp\left[-\left(2n + 1\right)\pi\right]} \sin\left(\frac{\pi}{2 K\left(-1\right)} \left(2n + 1\right) u\right)
\]
being $K(-1)$ the complete elliptic integral of the first kind. This means that
\[
G\left(t\right) = 2^{\frac{1}{4}} \theta\left(t\right) \cdot
\frac{2 \pi}{i K\left(-1\right)} \sum_{n = 0}^\infty \left(-1\right)^n \frac{\exp\left[-\left( n + \frac{1}{2} \right) \pi\right]}{1 + \exp\left[-\left(2n + 1\right)\pi\right]} \sin\left(\frac{\pi}{2 K\left(-1\right)}\left(2n + 1\right) \frac{t}{2^{\frac{1}{4}}}\right),
\]
and we can identify a set of frequencies
\[
\omega_n = \frac{\pi}{2 K\left(-1\right)} \left(2n + 1\right) \frac{1}{2^{\frac{1}{4}}}
\]
that yield the leading contribution to the quantum spectrum of the theory that is the well-known $\phi^4$. This spectrum was firstly obtained with these techniques in \cite{Frasca:2005sx}.

In a similar way, one has from eq.~(\ref{eq:am}) that the spectrum has the peculiar form
\[
\omega_n=\frac{n\pi}{\sqrt{2}K(\sqrt{2})}
\]
and presents both a zero mode and unstable modes because $K(\sqrt{2})$ has also an imaginary part. This is the well-known sine-Gordon model. It has been discussed in \cite{Frasca2007}.

We can try to extend this to some other cases given in Tab.~\ref{tab1}. The only one amenable to a Fourier analysis is the $w^2$ case. One has
\begin{equation}\label{eq:wp}
G\left(t\right) = -\frac{1}{c} \theta\left(t\right) \wp\left(c t + c_1; 0, c_2\right)
\end{equation}
where $\wp(u;\omega_1,\omega_2)$ is the Weierstrass elliptic function. The Fourier series is given by \cite[Eq.~23.8.2]{DLMF}
\begin{equation}
    \wp\left(z\right) = -\frac{\eta_1}{\omega_1} + \frac{\pi^2}{4 \omega_1^2} \csc^2\left(\frac{\pi z}{2\omega_1}\right)
    - \frac{2\pi^2}{\omega_1^2} \sum_{n = 1}^\infty \frac{n \exp\left[i 2n \pi \frac{\omega_2}{\omega_1}\right]}{1 - \exp\left[i 2n \pi \frac{\omega_2}{\omega_1}\right]} \cos\left(\frac{n \pi z}{\omega_1}\right),
\end{equation}
where $\omega_1$ and $\omega_2$ are the periods of the Weierstrass function with the former real and the latter complex with a positive imaginary part. We have also
\begin{equation}
    \eta_1 = \frac{\pi^2}{2\omega_1} \left(\frac{1}{6} + \sum_{n = 1}^\infty \csc^2\left(\frac{n\pi\omega_2}{\omega_1}\right)\right).
\end{equation}
From eq.~(\ref{eq:wp}) we recognize that the parameter in the Weierstrass function yields the set of equations to solve to obtain the periods $\omega_1$ and $\omega_2$
\[
0 = \sum_{\left(n, m\right) \ne (0,0)} \frac{1}{\left(m \omega_1 + m \omega_2\right)^4}
\qquad
c_2 = \sum_{\left(n, m\right) \ne (0,0)} \frac{1}{\left(m \omega_1 + m \omega_2\right)^6}.
\]
Therefore, for $\omega_1 \in \mathbb{R}$ we get a discrete spectrum
\[
\epsilon_n = \frac{n \pi c}{\omega_1}
\]
with $n = 0$ excluded.

Other cases in Tab.~\ref{tab1} do not have a discrete spectrum.

\section{Related PDEs of Physical Interest and Generalized Variable Separation}

Using the generalized separation of variables \cite{Polyanin2012}, it is possible to show that the nonlinear PDEs of the form
\begin{equation}\label{nonlinwave}
\frac{\partial^\alpha \tilde{w}}{\partial t^\alpha} = \frac{\partial^2 \tilde{w}}{\partial x^2} + \tilde{N}\left(\frac{\partial \tilde{w}}{\partial x}, \tilde{w}, x, t\right), ~~ \alpha = 1, 2,
\end{equation}
are reduced to nonlinear ODEs of the form (\ref{gennonlin}) with a derivative dependent potential. Therefore, the modified NLG method can be applied for their approximation. Eq.~(\ref{nonlinwave}) arises in various areas of science including gravity, quantum field theory, engineering and fluid mechanics (describing, as a rule, nonlinear wave phenomena in solids or fluids \cite{Drumheller1998}), biology \cite{Murray2002}, and many others. Particular cases include the Allen–Cahn, generalized Burgers, Duffing–van der Pol, Emden, Fisher, Kakutani-Kawahara, Scr{\"o}dinger and other second order nonlinear equations.

To demonstrate the reduction procedure, let us consider the one-dimensional nonlinear wave equation
\[
\frac{\partial^2 \tilde{w}}{\partial t^2} = \alpha \frac{\partial}{\partial x}\left[\exp\left[\lambda x\right] \frac{\partial \tilde{w}}{\partial x}\right] + \tilde{N}\left(\frac{\partial \tilde{w}}{\partial x}, \tilde{w}, x, t\right),
\]
with a general non-linearity $\tilde{N}$ and arbitrary real parameters $\alpha$ and $\lambda$. This equation allows separation of variables \cite{Polyanin2012}
\begin{equation}\label{gensep}
\chi^2 = a_1\left[ \frac{\exp\left[- \lambda x\right]}{\alpha \lambda^2} - \frac{\left(t + a_2\right)^2}{4} \right],
\end{equation}
with arbitrary constants $a_1$ and $a_2$. The reduced ODE reads as
\[
\frac{d^2 w}{d \chi^2} + \frac{4}{a_1} N\left(\frac{d w}{d \chi}, w, \chi\right) = 0.
\]
Here the symbols without tilde denote the corresponding quantities in the new variable $\chi$.

Furthermore, the additive separation of variables
\[
\tilde{w}\left(x, t\right) = \psi\left(x\right) + \varphi\left(t\right)
\]
reduces the nonlinear wave equation
\[
\frac{\partial^2 \tilde{w}}{\partial t^2} = \frac{\partial^2 \tilde{w}}{\partial x^2} + N_x\left(\frac{\partial \tilde{w}}{\partial x}, x\right) + N_t\left(\frac{\partial \tilde{w}}{\partial t}, t\right)
\]
with generic non-linearities $N_x$ and $N_t$ to the following system of two uncoupled ODEs \cite{Polyanin2012}:
\begin{equation}\label{reducedwave}
\frac{d^2 \psi}{d x^2} + N_x\left(\frac{d \psi}{d x}, x\right) = C, ~~ 
\frac{d^2 \varphi}{d t^2} - N_t\left(\frac{d \varphi}{d t}, t\right) = C,
\end{equation}
where $C$ is an arbitrary constant.

\section{Representation of the nonlinear Green's function for PDE with a non-linearity in the first derivative}

In this section we address the problem of representation of the nonlinear Green's function of second order differential equations of the form
\begin{equation}\label{nonlinODE}
    \frac{d^2 w}{d t^2} + N\left(\frac{d w}{d t}, w, t\right) = f\left(t\right), ~~ t > 0,
\end{equation}
complemented by some Cauchy conditions, where $f$ is a given function. As we saw in the previous section, quite general classes of PDEs can be reduced to (\ref{nonlinODE}). For the sake of simplicity, we limit the consideration to the one-dimensional case. 

General representation of the nonlinear Green's function in the case when
\[
N = N\left(w, t\right)
\]
has been studied in \cite{Frasca2018}. Recall its main result.

\begin{theorem}[\cite{Frasca2018}]
Assume that
\begin{equation}\label{Ncond}
N\left(\theta \cdot w, t\right) = \theta\left(t\right) \cdot N\left(w, t\right)
\end{equation}
holds. Then, the nonlinear Green's function of (\ref{gennonlin}) admits the following representation:
\[
G\left(t\right) = \theta\left(t\right) w_0\left(t\right),
\]
where $w_0$ is the general solution of the following Cauchy problem:
\[
\frac{d^2 w_0}{d t^2} + N\left(w_0, t\right) = 0, ~~ t > 0,
\]
\[
w_0\left(0\right) = 0, ~~ \frac{d w_0}{d t}\bigg|_{t = 0} = s.
\]
\end{theorem}

In other words, if the non-linearity possesses the multiplicativity property (\ref{Ncond}), then the nonlinear Green's is represented in terms of the homogeneous solution. This is an important result from applications point of view, since it allows to construct Green's function of a non-homogeneous equation by means of the general solution of its homogeneous part. It is noteworthy that the most part of handbooks containing exact solutions of nonlinear differential equations (see, for instance, \cite{Polyanin2017, Polyanin2012}) provide with their homogeneous solutions. At this, solution of nonlinear non-homogeneous equations is a very challenging topic in modern theory of nonlinear differential equations.

Then, the main result of this paper is the following theorem.
\begin{theorem}\label{theorderiv}
Let us consider a non-linear PDE with non-lineatiries in the first derivatives. If the following multiplicativity relation holds:
\begin{equation}\label{Ncondderiv}
N\left(\theta \cdot \frac{d w}{d t}, \theta \cdot w, t\right) = \theta\left(t\right) \cdot N\left(\frac{d w}{d t}, w, t\right),
\end{equation}
then, the nonlinear Green's function of (\ref{nonlinODE}) admits the following representation:
\begin{equation}\label{nonlinGreen}
    G\left(t\right) = \theta\left(t\right) w_0\left(t\right)
\end{equation}
where $w_0$ is the general solution of the following \emph{homogeneous} Cauchy problem:
\begin{equation}\label{nonlinODEhom}
\frac{d^2 w_0}{d t^2} + N\left(\frac{d w_0}{d t}, w_0, t\right) = 0, ~~ t > 0,
\end{equation}
\begin{equation}\label{nonhomCauchy}
w_0\left(0\right) = 0, ~~ \frac{d w_0}{d t}\bigg|_{t = 0} = s.
\end{equation}
\end{theorem}

\begin{proof}

Consider the Cauchy problem (\ref{nonlinODEhom}), (\ref{nonhomCauchy}). Multiplying both sides of (\ref{nonlinODEhom}) by $\theta$ and making use of (\ref{Ncondderiv}), we derive
\[
\theta\left(t\right) \frac{d^2 w_0}{d t^2} + N\left(\theta \cdot \frac{d w_0}{d t}, \theta \cdot w_0, t\right) = 0.
\]
On the other hand, in the sense of distributions we have
\[
\frac{d \theta}{d t} = \delta\left(t\right),
\]
\[
\frac{d \left(\theta \cdot w_0\right)}{d t} = w_0\left(0\right) \delta\left(t\right) + \theta\left(t\right) \cdot \frac{d w_0}{d t},
\]
\[
\frac{d^2 \left(\theta \cdot w_0\right)}{d t^2} = w_0\left(0\right) \delta'\left(t\right) + \frac{d w_0}{d t}\bigg|_{t = 0} \cdot \delta\left(t\right) + \theta\left(t\right) \cdot \frac{d^2 w_0}{d t^2}.
\]
By virtue of (\ref{nonhomCauchy}) we have that
\[
\theta\left(t\right) \cdot \frac{d w_0}{d t} = \frac{d \left(\theta \cdot w_0\right)}{d t}, ~~ \theta\left(t\right) \cdot \frac{d^2 w_0}{d t^2} = \frac{d^2 \left(\theta \cdot w_0\right)}{d t^2} - s \delta\left(t\right).
\]
Substituting these into (\ref{nonlinODEhom}), we obtain that the function $G$ (\ref{nonlinGreen}) satisfies
\[
\frac{d^2 G}{d t^2} + N\left(\frac{d G}{d t}, G, t\right) = s \delta\left(t\right).
\]
Thus, it is the nonlinear Green's function of (\ref{nonlinODE}).

\end{proof}

The particular solutions obtained in \cite{Frasca2018} for the case (\ref{Ncond}) and all reasonable combinations of them can be involved in this case as well.

\section{Applications}

In this section we consider some new non-linearities that allow to derive the corresponding Green's function explicitly in the form of (\ref{nonlinGreen}).

\subsection{Generalized Burgers' Equation}

Consider the generalized Burgers' equation
\[
\frac{\partial \tilde{w}}{\partial t} = \frac{\partial^2 \tilde{w}}{\partial x^2} - \tilde{w} \left(\frac{\partial \tilde{w}}{\partial t}\right)^2.
\]
This equation admits a traveling wave solution $\tilde{w}\left(x, t\right) = w\left(x - v t\right) := w\left(\chi\right)$, $v = const$, determined from the second order ODE
\begin{equation}\label{reducedBurg}
\frac{d^2 w}{d \chi^2} + v \frac{d w}{d \chi} - w \left(\frac{d w}{d \chi}\right)^2 = 0.
\end{equation}

Obviously, the nonlinear term
\[
N\left(\frac{d w}{d \chi}, w, \chi\right) = v \frac{d w}{d \chi} - w \left(\frac{d w}{d \chi}\right)^2
\]
satisfies (\ref{Ncondderiv}). The general solution of (\ref{reducedBurg}), (\ref{nonhomCauchy}) is found explicitly as follows:
\[
w_0\left(\chi\right) = \sqrt{2} \operatorname{erf}^{-1}\left[ \sqrt{\frac{2}{\pi}} \frac{s}{v} \left[1 - \exp \left(- v \chi\right) \right] \right].
\]
Therefore, according to Theorem \ref{theorderiv}, the nonlinear Green's function is found to have the following form:
\[
\tilde{G}\left(x, t\right) = \theta\left(x - v t\right) w_0\left(x - v t\right) = \sqrt{2} ~ \theta\left(x - v t\right) \cdot \operatorname{erf}^{-1}\left[ \sqrt{\frac{2}{\pi}} \frac{s}{v} \left[1 - \exp \left(- v x + v^2 t\right) \right] \right],
\]
and its traveling wave solution is represented in the form of (\ref{shorttime}).

\subsection{Nonlinear heat equation with power temperature gradient}

The nonlinear heat equation with nonlinear temperature gradient
\begin{equation}\label{heatPDE}
\frac{\partial \tilde{w}}{\partial t} = \frac{\partial^2 \tilde{w}}{\partial x^2} + \left(\frac{\partial \tilde{w}}{\partial x}\right)^n, ~~ n = 2, 3, 4, \dots,
\end{equation}
admits the traveling wave solution $\tilde{w}\left(x, t\right) = w\left(x - v t\right) := w\left(\chi\right)$, $v = const$, determined from
\begin{equation}\label{heatred}
\frac{d^2 w}{d \chi^2} + v \frac{d w}{d \chi} + \left(\frac{d w}{d \chi}\right)^n = 0.
\end{equation}

In this case also the nonlinear term
\[
N\left(\frac{d w}{d \chi}, w, \chi\right) = v \frac{d w}{d \chi} + \left(\frac{d w}{d \chi}\right)^n
\]
satisfies (\ref{Ncondderiv}). The general solution of (\ref{heatred}), (\ref{nonhomCauchy}) reads as follows:
\[
w_0\left(\chi\right) = H\left(s\right) - H\left(G_0\left(\chi\right)\right)
\]
where
\[
H\left(u\right) = \frac{u}{s} \cdot ~ _2F_1\left(1, \frac{1}{n - 1}; \frac{n}{n - 1}; - \frac{u^{n - 1}}{v} \right),
\]
$_2F_1$ is the Gauss hypergeometric function,
\[
G_0\left(\chi\right) = g^{-1}\left(g\left(s\right) - \chi\right), ~~ g\left(u\right) = \frac{1}{\left(n - 1\right) v} \ln\frac{1}{1 + v u^{1 - n}},
\]
and the superscript $-1$ denotes the inverse function.

Then, according to Theorem \ref{theorderiv}, the nonlinear Green's function of (\ref{heatPDE}) is found explicitly as follows:
\[
G\left(x, t\right) = \theta\left(x - v t\right) w_0\left(x - v t\right) = \theta\left(x - v t\right) \left[ H\left(s\right) - H\left(G_0\left(x - v t\right)\right) \right].
\]

In particular, if $n = 2$, we derive
\[
\tilde{G}\left(x, t\right) = \theta\left(x - v t\right) \left[ - v x + v^2 t + \ln\left[ \frac{s + v}{s} \exp\left(x - v t\right) - 1\right] \right].
\]

\subsection{Wave Equation with Nonlinear Damping}

It is easy to verify that the wave equation with nonlinear damping
\begin{equation}\label{wavePDE}
\frac{\partial^2 \tilde{w}}{\partial t^2} + \alpha \left(\frac{\partial \tilde{w}}{\partial t}\right)^n = c^2 \frac{\partial^2 \tilde{w}}{\partial x^2}, ~~ n = 3, 4, \dots,
\end{equation}
admits the traveling wave solution $\tilde{w}\left(x, t\right) = w\left(x - v t\right) = w\left(\chi\right)$, $v = const$, satisfying
\begin{equation}\label{wavered}
a \frac{d^2 w}{d \chi^2} + \alpha \left(\frac{d w}{d \chi}\right)^n = 0,
\end{equation}
where $a = c^2 - v^2$, $\alpha > 0$.

In this case also, the nonlinear term satisfies (\ref{Ncondderiv}). The general solution of (\ref{wavered}), (\ref{nonhomCauchy}) reads as follows:
\[
w_0\left(\chi\right) = \frac{a}{\alpha \left(n - 2\right)} \left[ -s^{2 - n} + \left( s^{1 - n} + \frac{\alpha \left(n - 1\right)}{a} \chi \right)^{1 + \frac{1}{1 - n}} \right].
\]
Therefore, according to Theorem \ref{theorderiv}, the nonlinear Green's function of (\ref{wavePDE})is determined as follows:
\[
\tilde{G}\left(x, t\right) = \frac{a}{\alpha \left(n - 2\right)} \theta\left(x - v t\right) \left[ -s^{2 - n} + \left( s^{1 - n} + \frac{\alpha \left(n - 1\right)}{a} \left(x - v t\right) \right)^{1 + \frac{1}{1 - n}} \right].
\]

\section{Numerical Quantification of the Approximation Error: Generalized Burgers' Equation}

Now we consider some specific examples of nonlinear PDEs borrowed from existing references. The numerical error in computations is quantified by means of the logarithmic error function
\[
\operatorname{Er}\left(\chi\right) = \log_{10}\left|w_{\rm Green's}\left(\chi\right) - w_{\rm MoL}\left(\chi\right)\right|,
\]
measuring the mismatch between the nonlinear Green's solution $w_{\rm Green's}$ and the numerical solution $w_{\rm MoL}$ obtained by means of the method of lines in degrees of $10$. For the sake of simplicity, we restrict the consideration by the first order approximation (\ref{Frascaapprox}). The numerical error is minimized with respect to $s_1$ and $s_2$ simultaneously. The error can be reduced further by adding more terms of the short time expansion (\ref{shorttime}).

In order to test the proposed approach, we consider the numerical solution of (\ref{reducedBurg}) for different source functions $f$. Figs. \ref{fig1}--\ref{fig6}, as well as Tab. \ref{tab2} show that the approximate solution is efficiently close to the exact one. In Tab. \ref{tab2} we also bring the values of the scale parameters above.

\begin{figure}[H]
\centerline{\includegraphics[width = 3in]{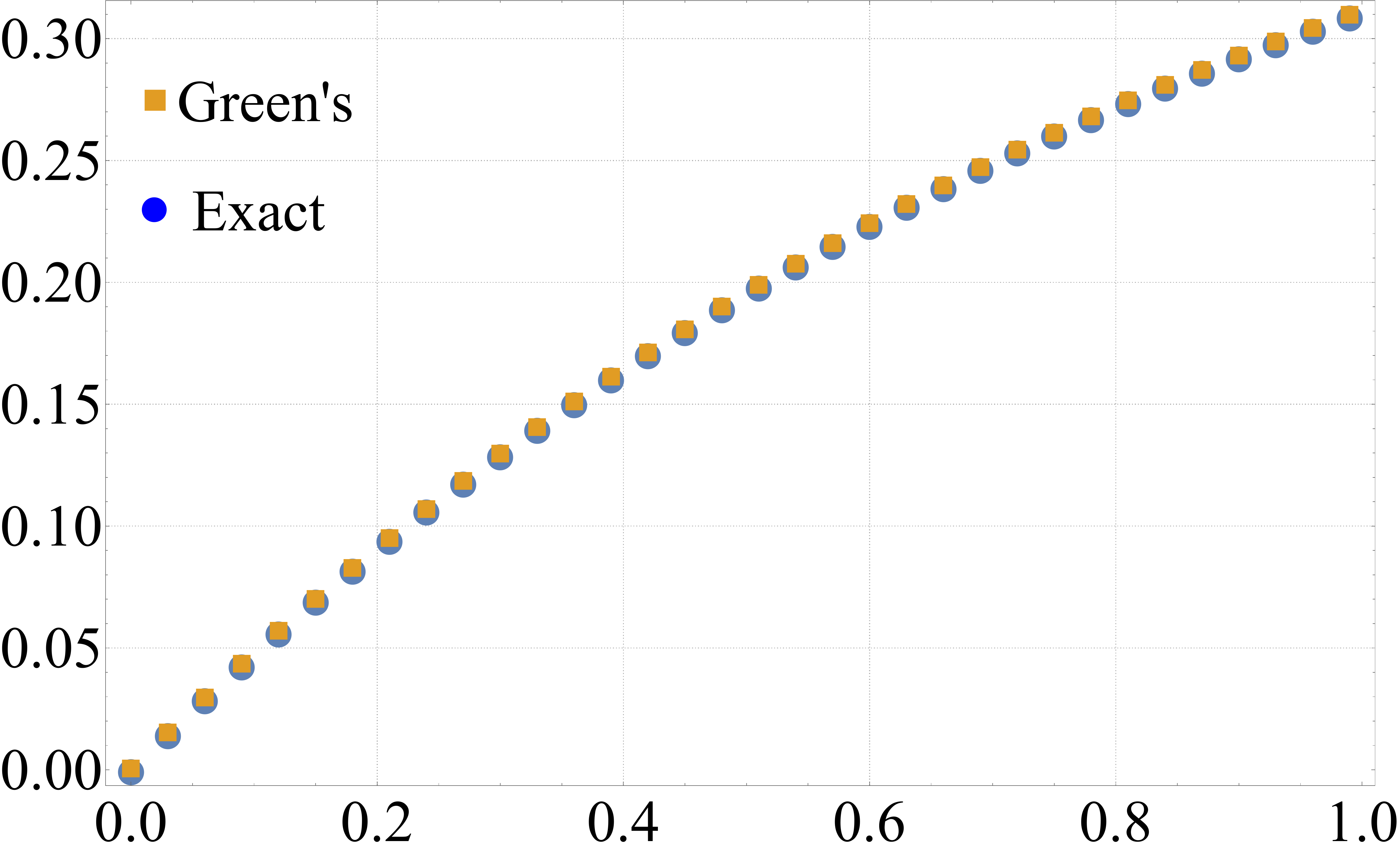} ~ \includegraphics[width=3.05in]{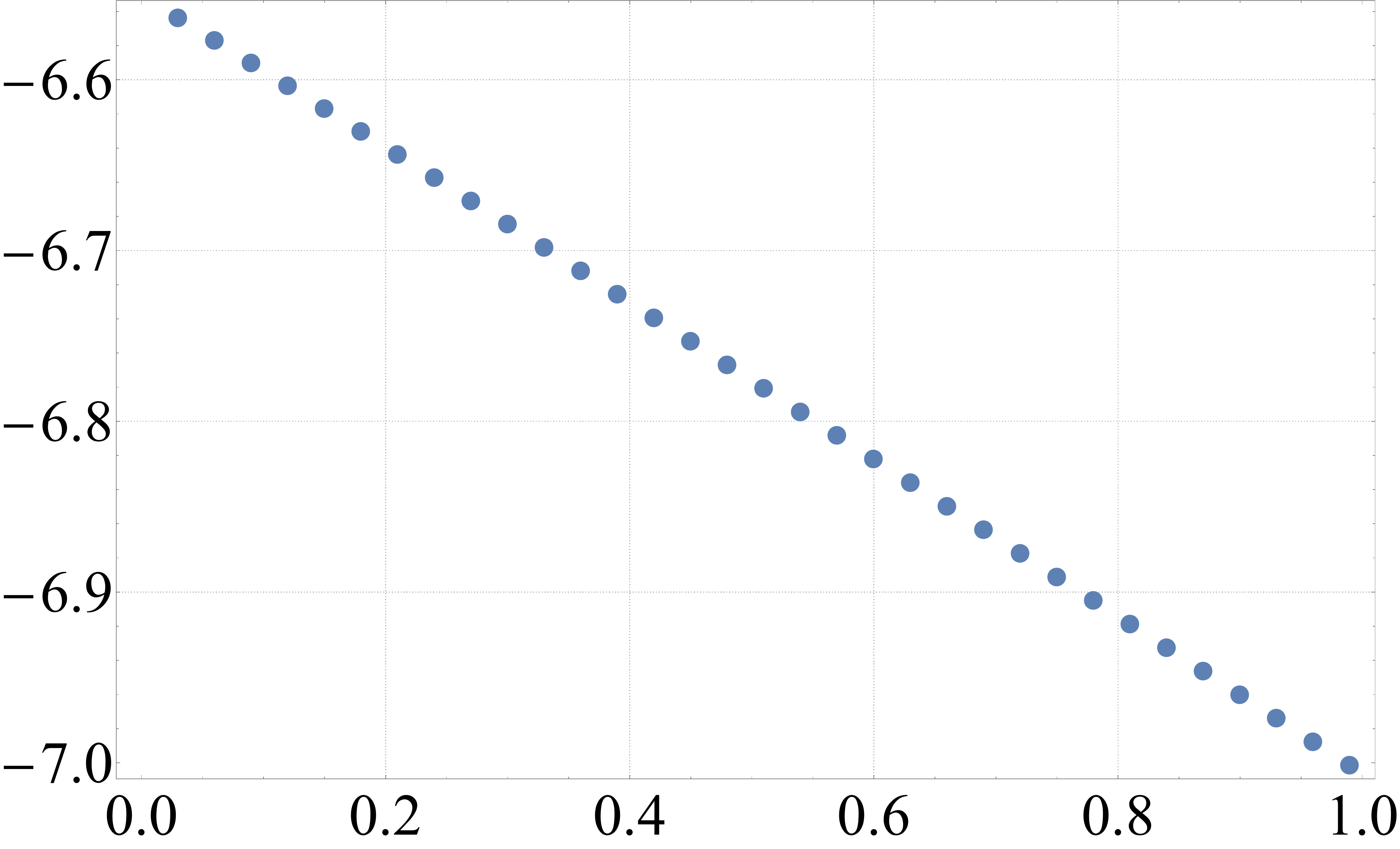}}
\vspace*{8pt}
\caption{Discrete plot of exact and approximate solutions (left) and $\operatorname{Er}$ (right) for $f\left(\chi\right) = \delta\left(\chi\right)$: generalized Burgers' equation}
\label{fig1}
\end{figure}

\begin{figure}[H]
\centerline{\includegraphics[width = 3in]{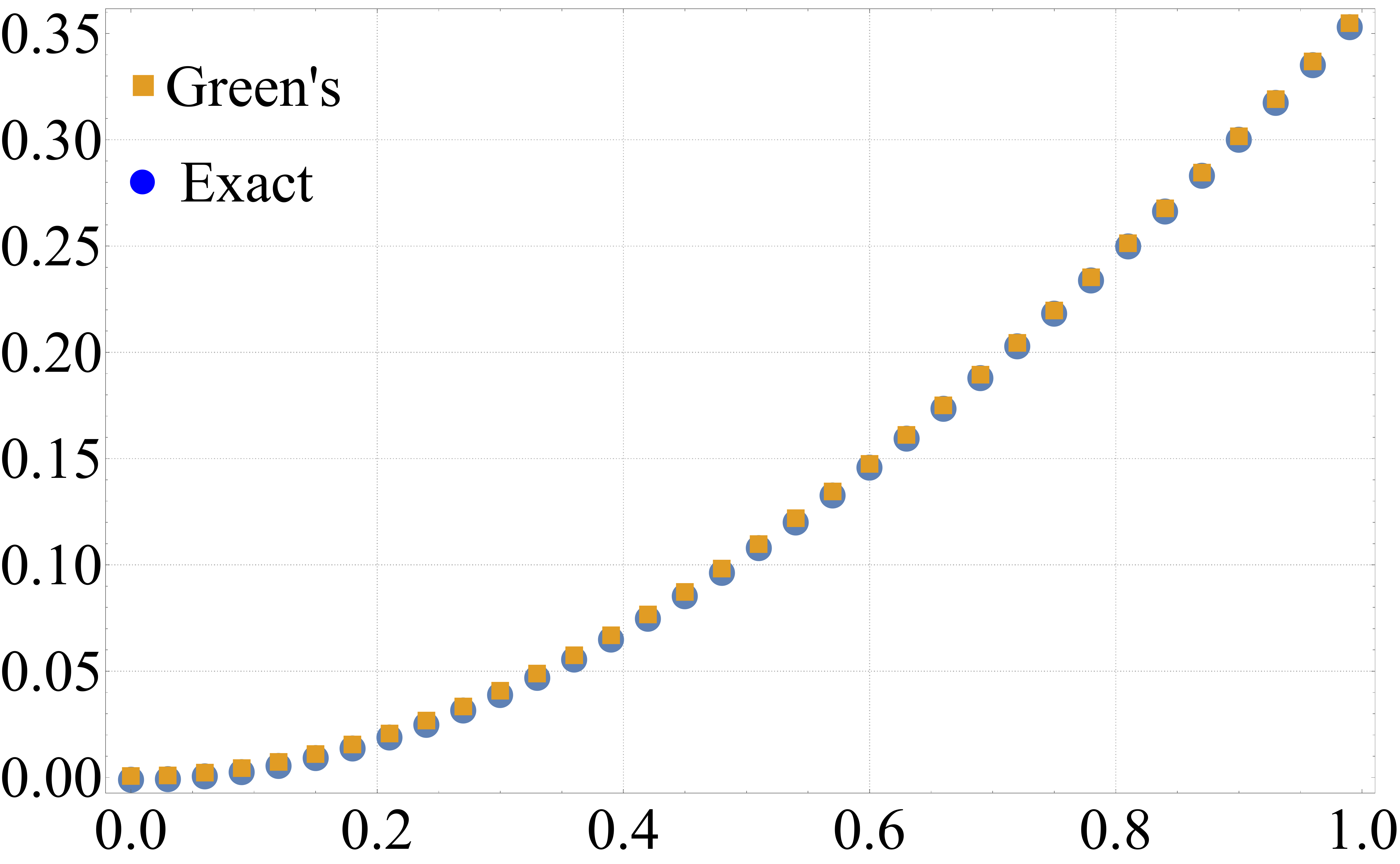} ~ \includegraphics[width=3in]{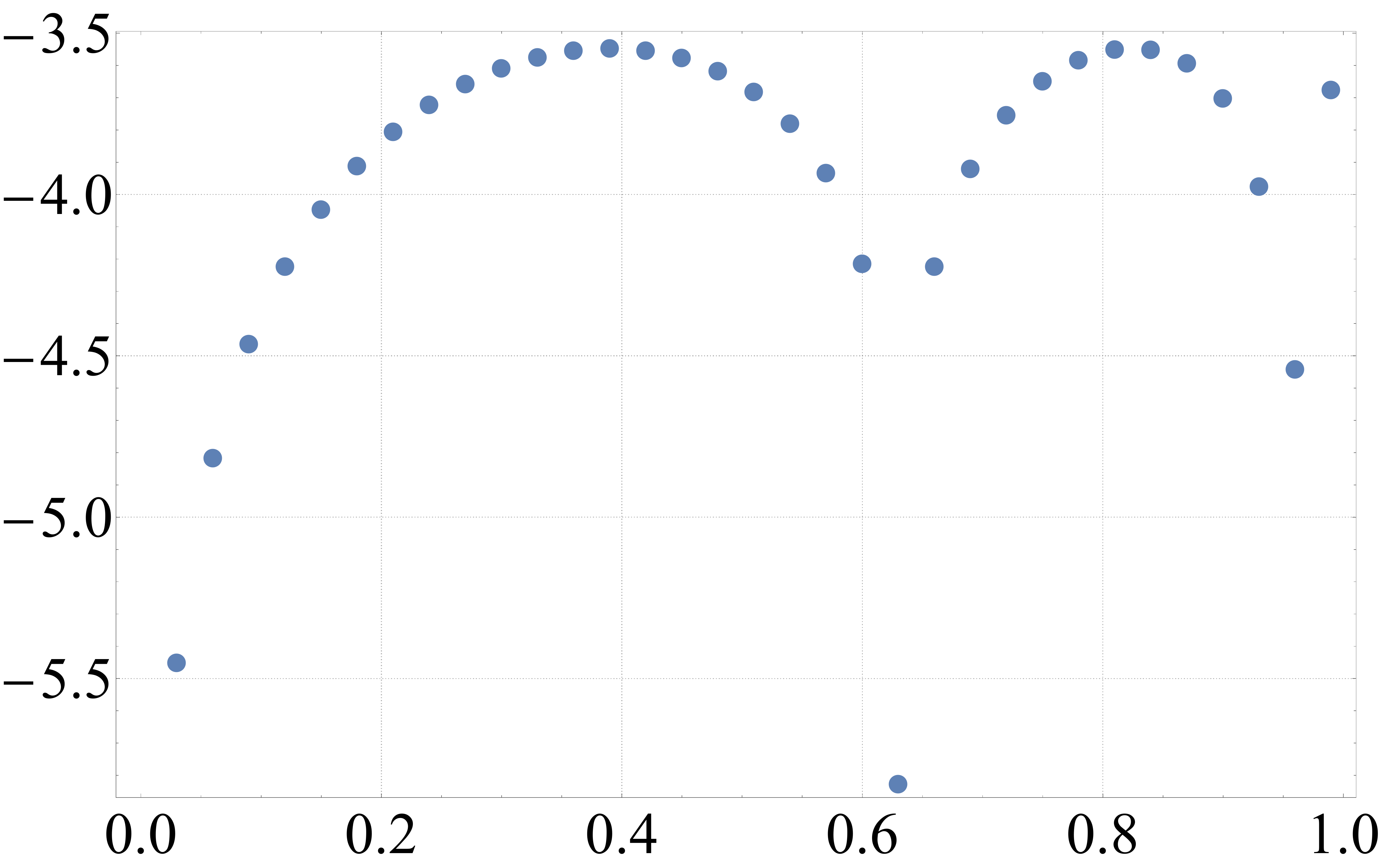}}
\vspace*{8pt}
\caption{Discrete plot of exact and approximate solutions (left) and $\operatorname{Er}$ (right) for $f\left(\chi\right) = \theta\left(\chi\right)$: generalized Burgers' equation}
\label{fig2}
\end{figure}

\begin{figure}[H]
\centerline{\includegraphics[width = 3in]{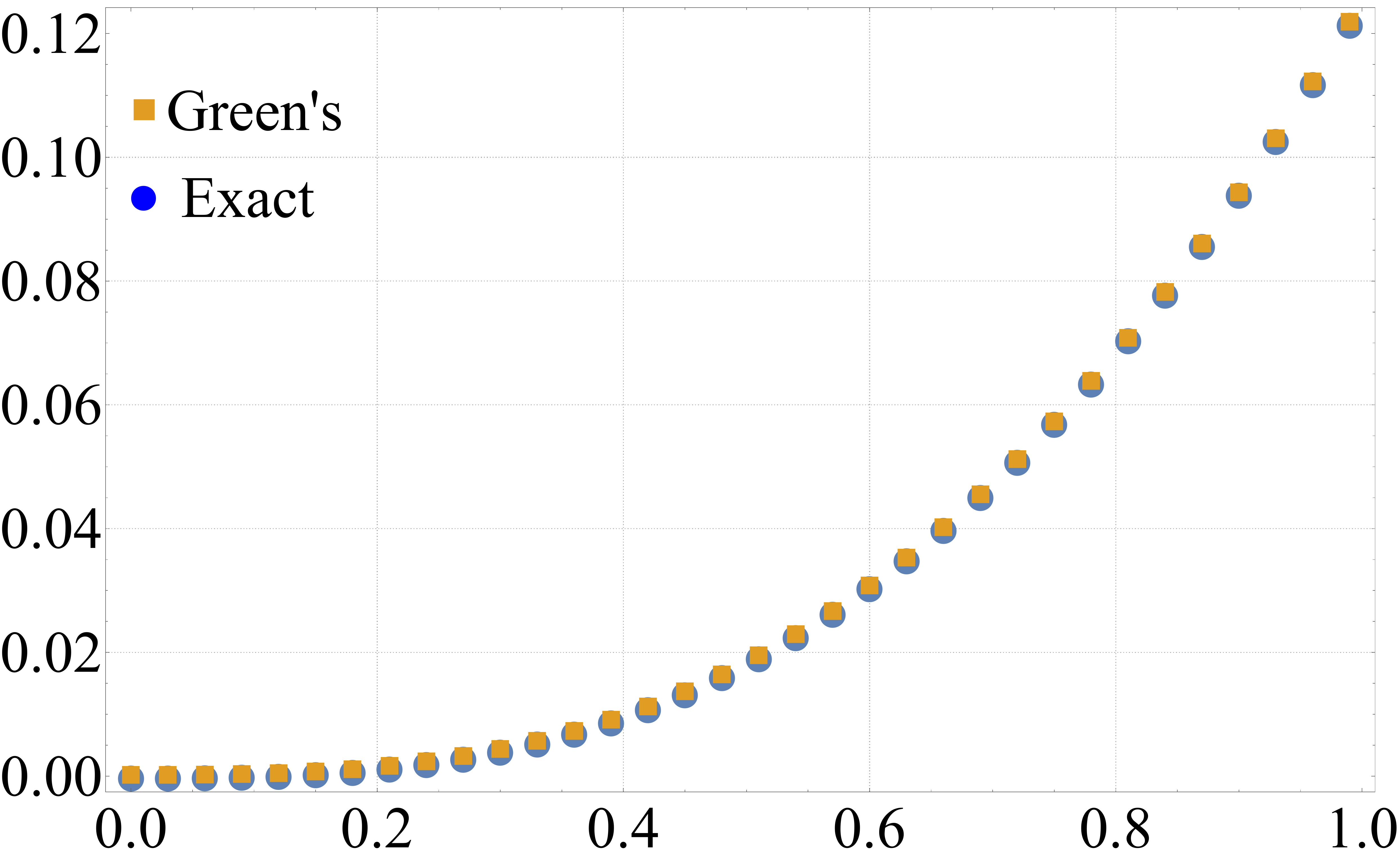} ~ \includegraphics[width=3in]{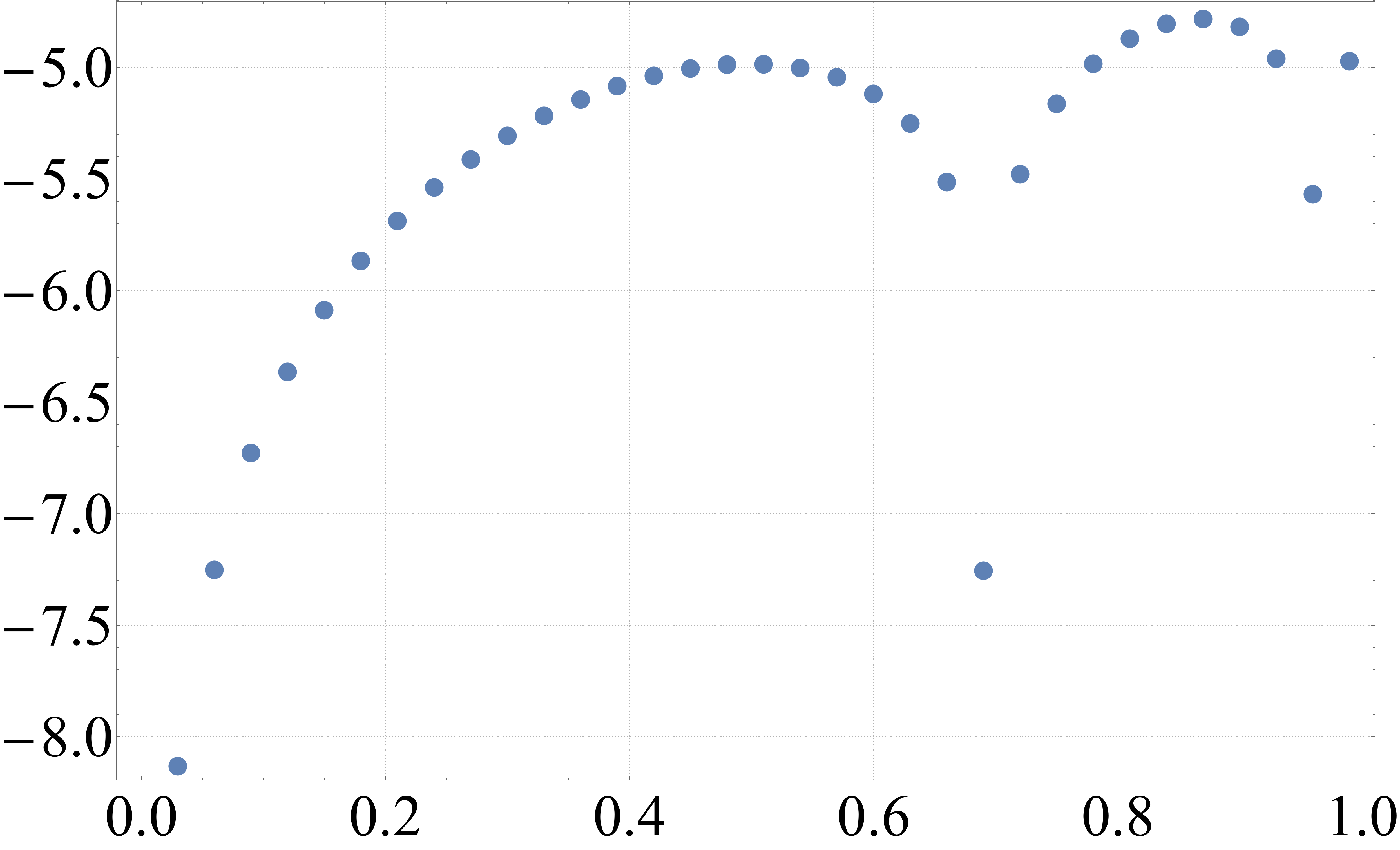}}
\vspace*{8pt}
\caption{Discrete plot of exact and approximate solutions (left) and $\operatorname{Er}$ (right) for $f\left(\chi\right) = \sin\left(\chi\right)$: generalized Burgers' equation}
\label{fig3}
\end{figure}

\begin{figure}[H]
\centerline{\includegraphics[width = 3in]{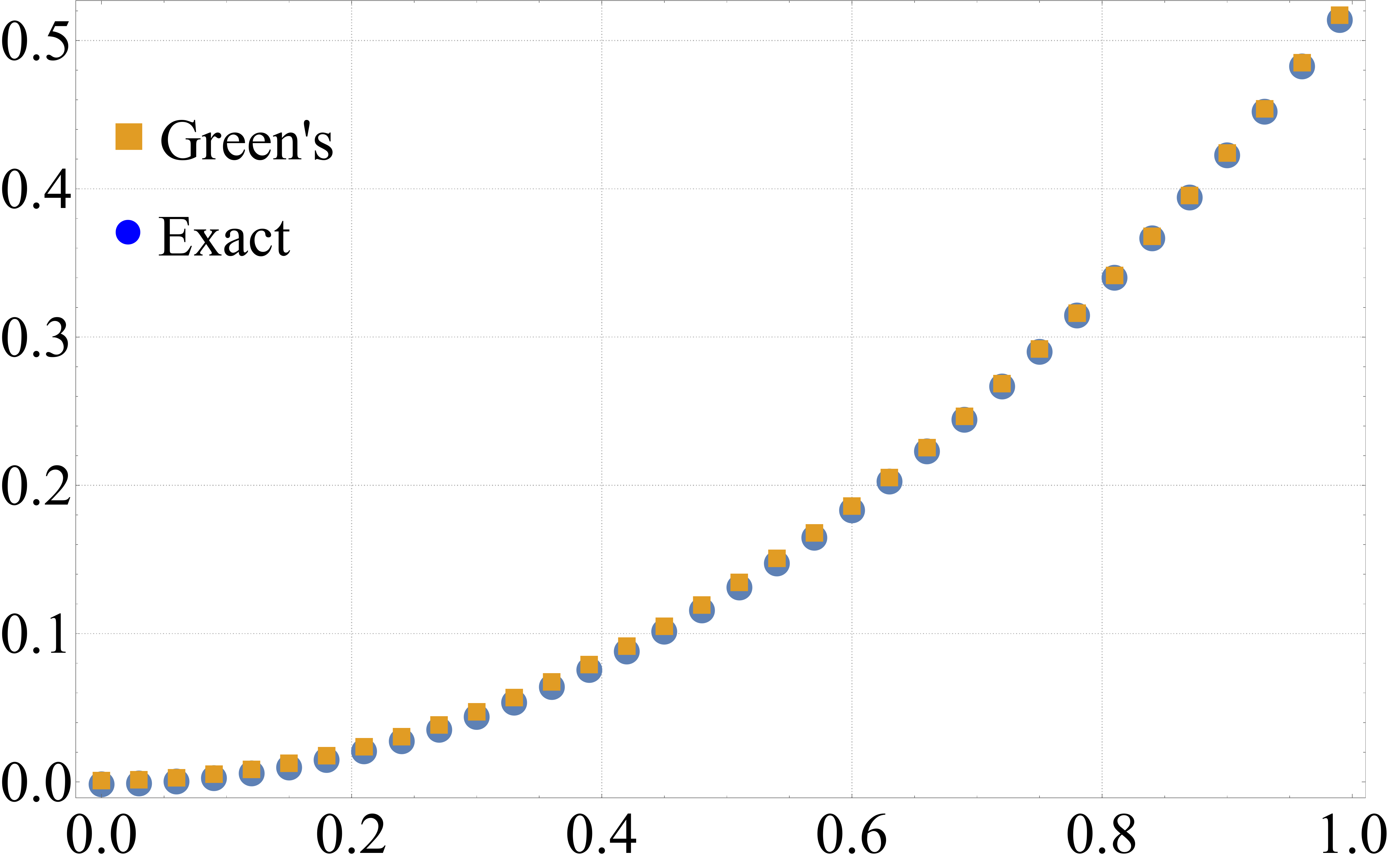} ~ \includegraphics[width=3in]{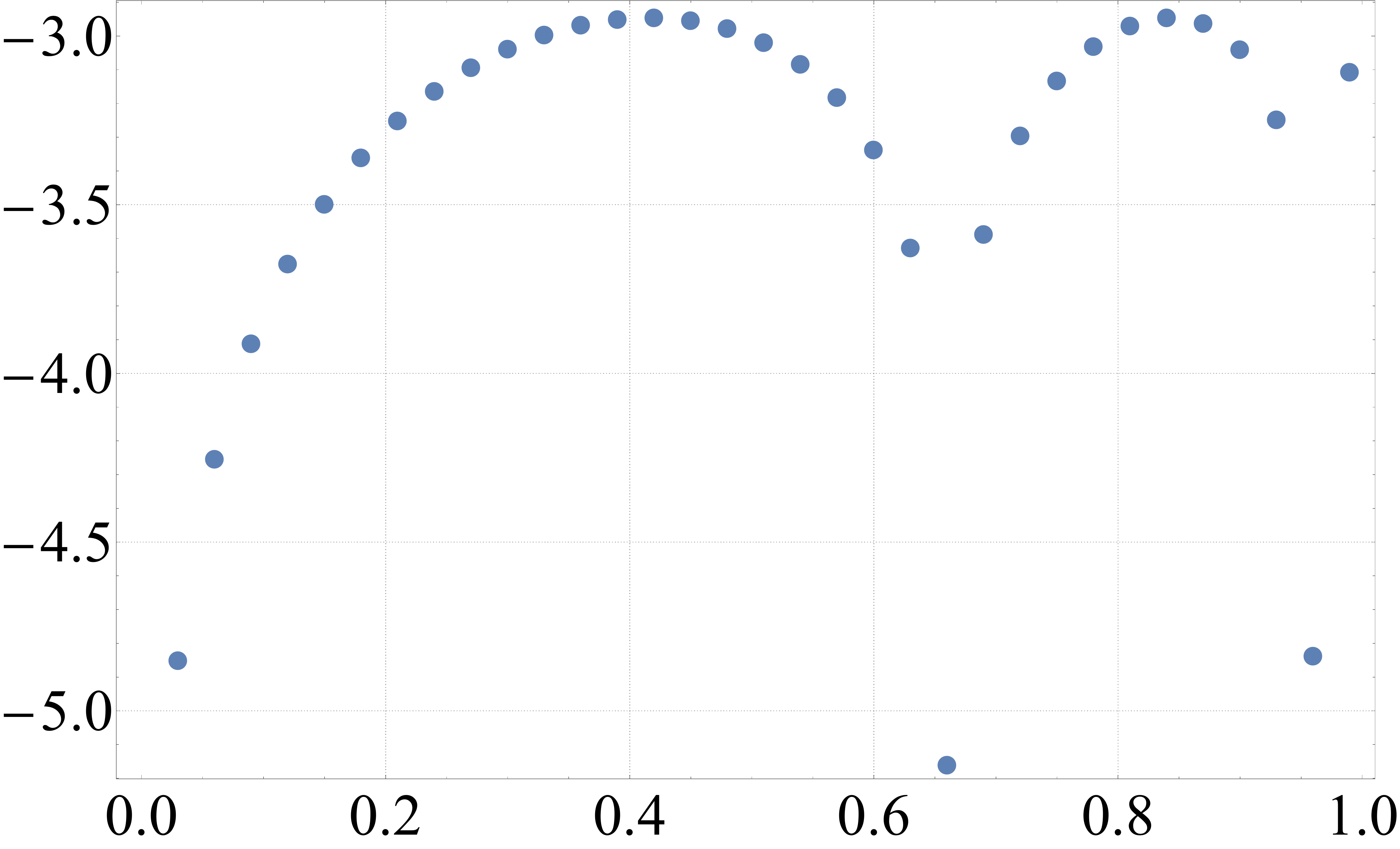}}
\vspace*{8pt}
\caption{Discrete plot of exact and approximate solutions (left) and $\operatorname{Er}$ (right) for $f\left(\chi\right) = \exp\left(\chi\right)$: generalized Burgers' equation}
\label{fig4}
\end{figure}

\begin{figure}[H]
\centerline{\includegraphics[width = 3in]{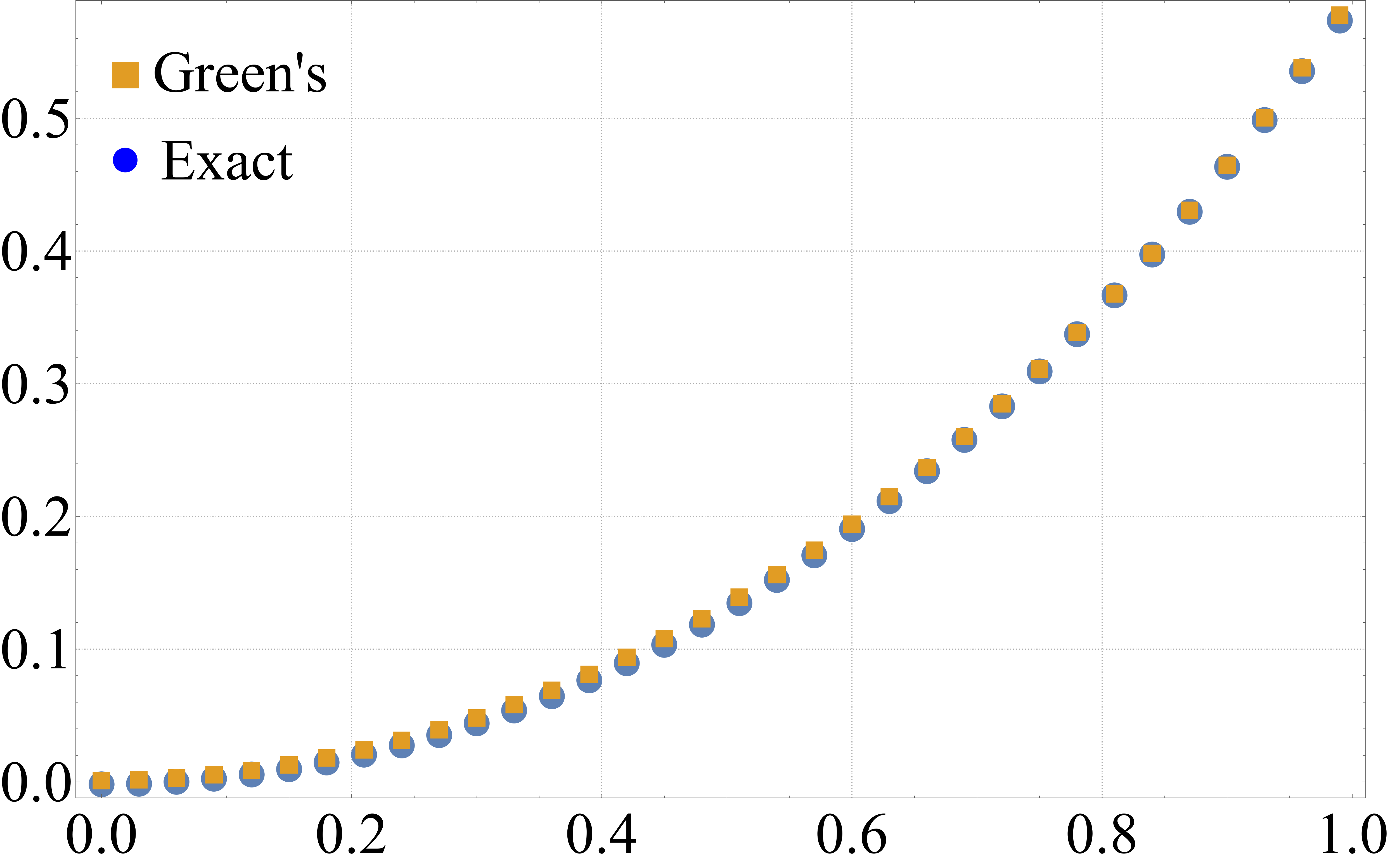} ~ \includegraphics[width=3in]{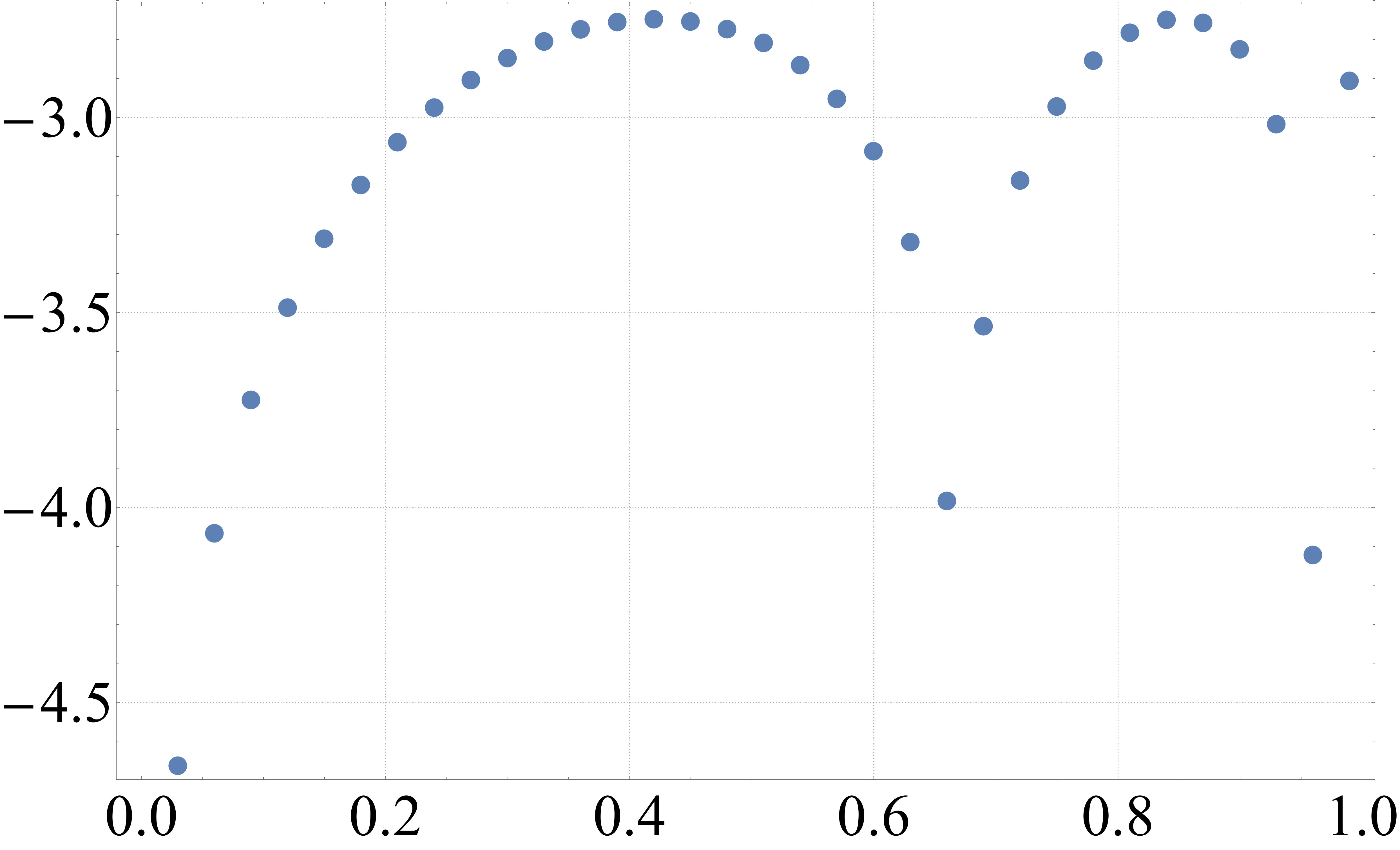}}
\vspace*{8pt}
\caption{Discrete plot of exact and approximate solutions (left) and $\operatorname{Er}$ (right) for $f\left(\chi\right) = 1 + \chi + \chi^2 + \chi^3$: generalized Burgers' equation}
\label{fig5}
\end{figure}

\begin{figure}[H]
\centerline{\includegraphics[width = 3in]{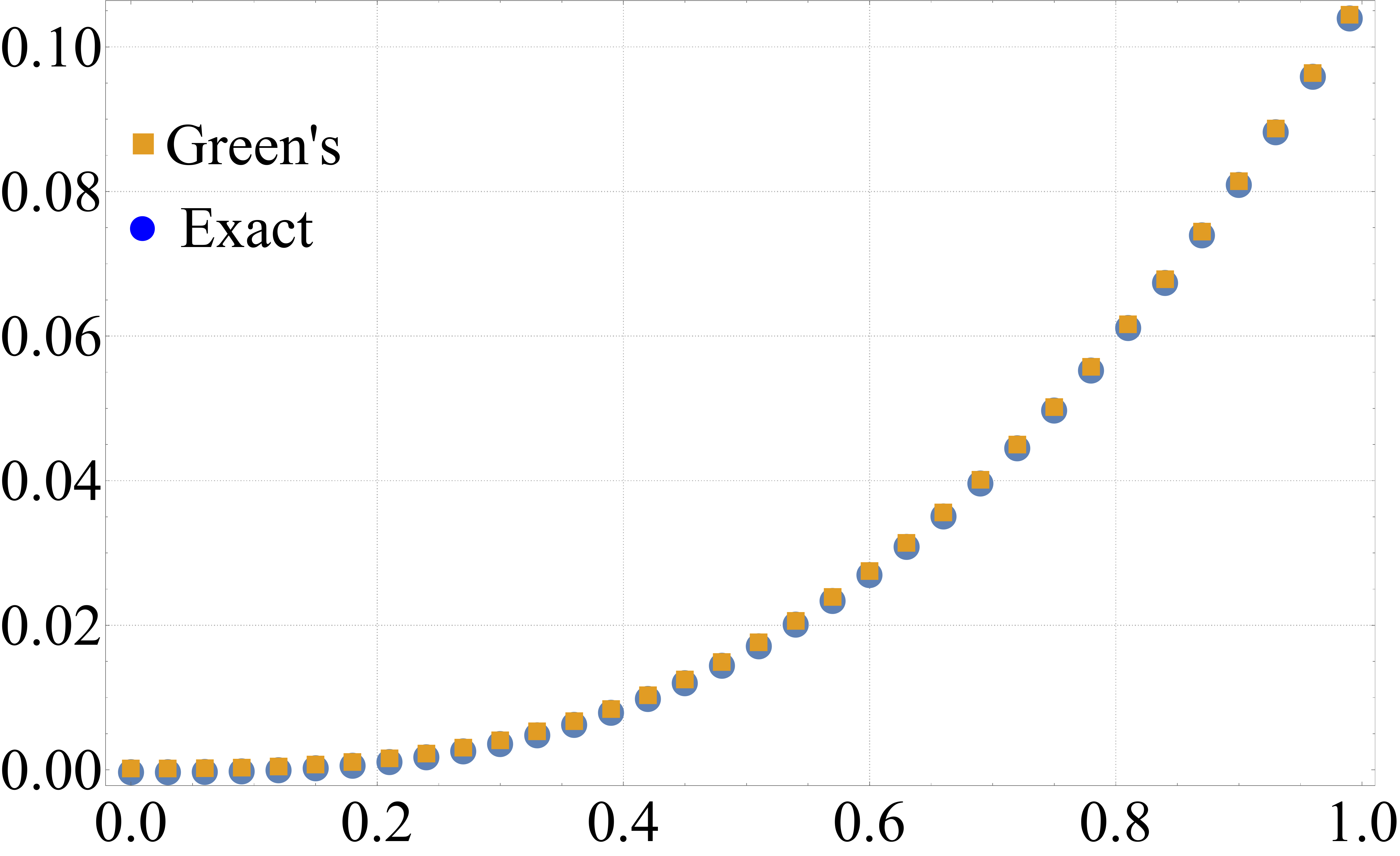} ~ \includegraphics[width=3in]{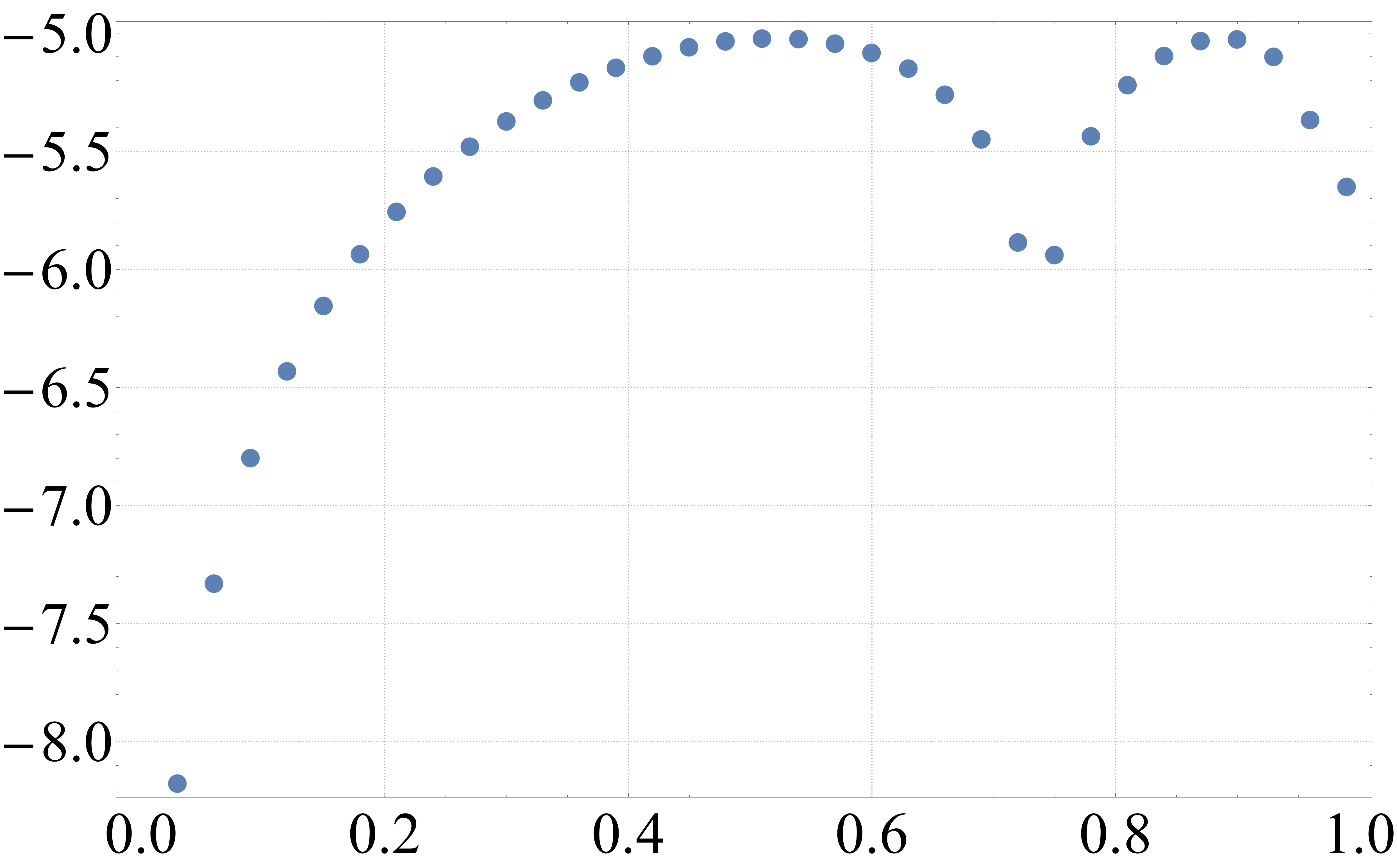}}
\vspace*{8pt}
\caption{Discrete plot of exact and approximate solutions (left) and $\operatorname{Er}$ (right) for $f\left(\chi\right) = \ln\left(1 + \chi\right)$: generalized Burgers' equation}
\label{fig6}
\end{figure}

\begin{table}[H]
\centering
\begin{tabular}{ c c c c c}
 $f\left(\chi\right)$ & $\min\operatorname{Er}$ & $\max\operatorname{Er}$ & $s_1$ & $s_2$ \\ \hline
 
 $\delta\left(\chi\right)$ & $-7$ & $-6.575$ & $2$ & $1$ \\ \hline
 
 $\theta\left(\chi\right)$ & $-6$ & $-3.5$ & $1.2179$ & $1.658$ \\ \hline
 
 $\sin\left(\chi\right)$ & $-5.1$ & $-4.8$ & $0.72126$ & $2.7776$ \\ \hline
 
 $\exp\left(\chi\right)$ & $-5.25$ & $-3$ & $0.712$ & $2.9$ \\ \hline
 
 $1 + \chi + \chi^2 + \chi^3$ & $-4.6$ & $-2.8$ & $0.5826$ & $3.6015$ \\ \hline

 $\ln\left(1 + \chi\right)$ & $-8.1$ & $-5$ & $3.10108$ & $0.6459$ \\ \hline 
\end{tabular}
\caption{Minimal and maximal logarithmic errors of approximation for various source functions: generalized Burgers' equation}
\label{tab2}
\end{table}

\section*{Conclusions}

New explicitly integrable cases for derivation of nonlinear Green's function are considered in this paper. It is shown that combining the modified NLG method with the generalized separation of variables or traveling wave {sl ansatz}, it is possible to approximate the solutions of nonlinear PDEs having special interest in various branches of physics. It is shown how the knowledge of the nonlinear Green's function allows to study the spectrum of the nonlinear operator. Several particular cases describing diffusion (generalized Burgers' equation), heat conduction and wave phenomena are studied in details. A numerical analysis is performed to quantify the logarithmic error between the exact and approximate solutions for particular non-linearities and various source functions. Minimizing the error with respect to the scale parameters, a low-error approximation is derived. It is concluded that the approach can be effectively applied for quantitative and qualitative analysis of various nonlinear processes described by second order PDEs. The NLG method proves in this way to grant both a closed form analytic solutions to non-homogeneous non-linear equations and, eventually, an algorithm to treat them numerically. This will grant a lot of possible applications in several field of physics and engineering.

\end{document}